\documentclass[12pt,a4paper]{article}

\usepackage{cmap} 

\usepackage[margin=20mm]{geometry}
\usepackage[T1]{fontenc}
\usepackage[utf8]{inputenx}
\usepackage{ae}
\usepackage{graphicx}
\usepackage[american]{babel}
\usepackage{amsfonts}

\usepackage{amsthm}
\usepackage{amsmath}
\usepackage{amssymb}
\usepackage{wasysym}
\usepackage{paralist}
\usepackage{multirow}
\usepackage{pifont}
\usepackage{imakeidx}
\usepackage{arydshln}
\usepackage{hyperref}

\makeindex
\newtheorem{Theorem}{Theorem}[section]
\newtheorem{Def}[Theorem]{Definition}
\newtheorem{Lemma}[Theorem]{Lemma}

\newtheorem{Note}[Theorem]{Note}
\newtheorem{Example}[Theorem]{Example}

\newcommand{\E}{{\sf E}}
\newcommand{\eps}{{\varepsilon}}

\renewcommand{\P}{{\sf P}}
\newcommand{\R}{\mathbb R}
\newcommand{\Str}{\mathcal S}

\renewcommand{\:}{\colon}
\newcommand{\where}{\ |\ }
\newcommand{\bigwhere}{\ \big|\ }

\usepackage[table]{xcolor}
\usepackage{tcolorbox}

\definecolor{red}{rgb}{1,0,0}
\definecolor{redish}{rgb}{0.8,0,0}
\definecolor{green}{rgb}{0,0.55,0}
\definecolor{blue}{rgb}{0,0,1}
\definecolor{lblue}{rgb}{0.5,0,1}
\definecolor{grey}{rgb}{0.5,0.5,0.5}
\definecolor{lgrey}{rgb}{0.88,0.88,0.88}
\definecolor{orange}{rgb}{1, 0.7, 0}
\definecolor{lorange}{rgb}{1, 1, 0.5}

\newcommand{\cred}{\color{red}}
\newcommand{\cblue}{\color{blue}}
\newcommand{\cgreen}{\color{green}}

\title{
A Robust Efficient Dynamic Mechanism 
}
\author{Endre Csóka\\{\normalsize Alfréd Rényi Institute of Mathematics, Budapest, Hungary}}
\date{}

\begin{document}
\maketitle
\begin{abstract}
Athey and Segal introduced an efficient budget-balanced mechanism for a dynamic stochastic model with quasilinear payoffs and private values, using the solution concept of perfect Bayesian equilibrium (PBE) \cite{AtSe}.
We show that this implementation is not robust in multiple senses, especially for at least 3 agents. For example, we will show a generic setup where all efficient strategy profiles can be eliminated by iterative elimination of weakly dominated strategies.
Furthermore, this model used strong assumptions about the information of the agents, and the mechanism was not robust to the relaxation of these assumptions.
In this paper, we will show a different mechanism that implements efficiency under weaker assumptions and uses the stronger solution concept of ``efficient Nash equilibrium with guaranteed expected payoffs''.
\end{abstract}

\section{Introduction}


The Vickrey--Clarke--Groves (VCG) mechanism \cite{Vickrey,Clarke,Groves} established the existence of an efficient, incentive-compatible mechanism for a general class of static mechanism design problems with private values and quasilinear preferences. 
Subsequently, Arrow \cite{arrow1979property} 
and d'Aspremont and Gérard-Varet 
(AGV) \cite{d1979incentives} constructed an efficient, incentive-compatible mechanism in which the transfers were also budget-balanced, using the solution concept of Bayesian–Nash equilibrium, under the additional assumption that private information is independent across agents.
In dynamic mechanism design problems with private values, Bergemann and Välimäki \cite{BeVa} 
and Athey and Segal \cite{AtSe} 
defined dynamic extensions of the VCG and AGV mechanisms.

In this paper, we will refer to the Athey--Segal balanced Team Mechanism.
We will show a very generic setup with at least 3 agents where the truthful PBE is degenerate in multiple senses: unstable, and not robust to the idealized information structure. Then we will also show that a different mechanism implements efficiency in a very stable and robust way, including even collusion-resistance.
This mechanism coincides with the balanced Team Mechanism only for 2 agents and with no same-round randomness in the private types.

\subsection{The setup and the balanced Team Mechanism}\label{sec:recall-setup}

We assume that the Reader knows the setup and the balanced Team Mechanism, but we give a short summary about them.

\medskip

We have a set $N = \{1, 2, ..., n\}$ of agents with fixed initial types $\theta^N_0 \in \Theta^N$ that are publicly known even to the designer, so the mechanism depends on the initial types $\theta^N_0$.
There are $T$ number of rounds. In each round $t$, a public decision $x_t \in X$ is made by the designer, and each agent $i$ gets utility $u(x_t, \theta^i_t)$ where $\theta^i_t \in \Theta$ is his current type. Between consecutive rounds, the type of each agent $i$ is changed by a stochastic function of the previous-round type and the public decision.\footnote{We note that the full version of the model and the mechanism includes private decisions. This was included only in the manuscript version of the Athey--Segal paper \cite{AtSe-old}. As our counterexample works in this restricted model, it works in the full model, as well. Both versions of the paper also allow verifiable public states which we also skipped from our summary because on the one hand, these can be replaced by an extra agent who cannot lie, and on the other hand, we will not use them in our counterexamples.}
Formally, each $\theta_{t+1}^i$ is chosen from the probability distribution $\mu(x_t, \theta_t^i) \in \Delta(\Theta)$, these randomizations are independent (and $\mu$ is fixed as a rule of the game). The agents can report their types, and the designer can prescribe monetary transfers between the agents $y^i$ with $\sum_i y^i = 0$. Every agent $i$ maximizes his expected total utility plus transfer $\E\big(\sum_t u(x_t, \theta^i_t) + y^i\big)$. Our goal is to maximize the total expected utility.

The information of each agent consists of his own type history, the public reports, and the (implied) public decisions.
It is strictly assumed that the agents know nothing else.
For example, no agent is capable to reveal to other agents any information correlated with his own past.

\begin{Def}[{\bf balanced Team Mechanism}]
We fix a decision strategy $\xi: \{1, 2, ..., T\} \times \Theta^N \to X$ which would maximize total expected utility. In every round $t$, we ask each agent to make a report $\widehat{\theta}^i_t$ about his type. We make the public decision $\xi(t, \widehat{\theta}^N_t)$. For each report $\widehat{\theta}^i_t$, the designer calculates how it changes the total of the expected utilities of others, given the (reported) types of all agents in the previous round $\widehat{\theta}^N_{t-1}$.
Formally,
\begin{equation} \label{gammadef}
    \gamma^i_t = \sum_{j \in N \setminus\{i\}} \sum_{t' = t}^T  \E \Big( \upsilon^{j, \xi}_{t'}\big(\widehat{\theta}^i_t,\  \widehat{\theta}^{N \setminus \{i\}}_{t-1}\big) - \upsilon^{j, \xi}_{t'}\big(\widehat{\theta}^N_{t-1}\big) \Big),
\end{equation}
where $\upsilon^{j, \xi}_{t'}$ expresses the utility of $j$ in round $t'$ given the type profile, and assuming truthful reporting strategies and decision strategy $\xi$. Each agent $i$ gets this signed transfer $\gamma^i_t$ paid equally by the other agents, namely, the total transfer to $i$ is 
\begin{equation} \label{eq:y-def}
y^i = \sum\limits_{t=1}^T \Big( \gamma^i_t - \frac{1}{n-1} \sum\limits_{j \in N \setminus \{i\}} \gamma^j_t \Big).
\end{equation}
\end{Def}

\section{The balanced Team Mechanism is not robust} \label{sec:wrong-ex}

In this section, we are pointing out one of the issues with the mechanism, temporarily accepting the idealized information structure of the model.
Our simplest example reveals only one aspect of instability, and this setup will have degeneracies. But then we will show more difficult and generic setups which will show even more serious instabilities.

\subsection{Our simplest example for instability}

\begin{Example} \label{general-example}
We have a large number $n \ge 3$ of agents.
Two agents have active roles: \textnormal{\cblue Blue} and \textnormal{\cred Red}, and an agent \textnormal{\cgreen Green} has a passive role. The utilities of others are 0, so their only role is paying the transfers $\gamma$ to \textnormal{\cblue Blue} and \textnormal{\cred Red}.
The game consists of $K \ge 2$ rounds.
At the end of the last round, \textnormal{\cblue Blue} and \textnormal{\cred Red} both have a final type HIGH or LOW, and there is a public decision to be made: \textnormal{YES} or \textnormal{NO}.
If \textnormal{NO}, then the utility of every agent is 0.
If \textnormal{YES}, then the utility of $(\textnormal{\cblue Blue}, \textnormal{\cred Red}, \textnormal{\cgreen Green})$ is $({\cblue 1\textnormal{ or }4}, {\cred 1\textnormal{ or }4}, {\cgreen -6})$.
Whether $1 \textnormal{ or } 4$ depends on the last-round type of the player: $1$ with \textnormal{LOW} or $4$ with \textnormal{HIGH}.
Therefore, the efficient decision is \textnormal{YES} if and only if both agents have \textnormal{HIGH} type.
\end{Example}

\begin{Note}
Example~\ref{general-example} is essentially a single-round setup, but we consider earlier rounds in which nothing happens except that the types of the agents are updated. We could easily extend the setup to give some sense to the earlier rounds, see Appendix~\ref{app:roundsense}.
\end{Note}

For $K = 2$ rounds, we show the table of expected payoffs of {\cblue Blue} and {\cred Red} with three strategies (to be explained shortly) for each of them. The top left cell (with yellow background) is the truthful equilibrium. The strategy in the second row (and column) weakly dominates the first row (and column, respectively).

\begin{center}
{\renewcommand{\arraystretch}{1.3}
\begin{tabular}[center]{c||c|c|c|}

\multirow{2}{*}{$\E\big($payoffs$(${\cblue Blue}, {\cred Red}$)\big)$}
& \cellcolor{lgrey} 
& \cellcolor{lgrey} ${\cred \widehat{p}^{red}_t} = {\cred p^{red}_t}$ except:
& \cellcolor{lgrey} ${\cred \widehat{p}^{red}_1} = 0$; \ ${\cred \widehat{p}^{red}_2} = \big(${\bf if} \\
& \multirow{-2}{*}{\cellcolor{lgrey} ${\cred \widehat{p}^{red}_t} = {\cred p^{red}_t}$}
& \cellcolor{lgrey} ${\cblue \widehat{p}^{blue}_1} = 0\ {\cred \Rightarrow \widehat{p}^{red}_2} = 1$
& \cellcolor{lgrey} ${\cblue \widehat{p}^{blue}_1} = 0$ {\bf then} 1 {\bf else} ${\cred p^{red}_2}\big)$ \\
\hline
\hline
\cellcolor{lgrey} & \cellcolor{lorange} & & \\
\cellcolor{lgrey} \multirow{-2}{*}{${\cblue \widehat{p}^{blue}_t} = {\cblue p^{blue}_t}$}
& \cellcolor{lorange}\multirow{-2}{*}{$({\cblue 1},\ {\cred 1})$}
& \multirow{-2}{*}{$({\cblue 1},\ {\cred 1})$}
& \multirow{-2}{*}{$({\cblue 1},\ {\cred 1})$} \\
\hline
\cellcolor{lgrey} ${\cblue \widehat{p}^{blue}_t} = {\cblue p^{blue}_t}$ except:
& \multirow{2}{*}{$({\cblue 1},\ {\cred 1})$}
& \multirow{2}{*}{$({\cblue 1},\ {\cred 1})$}
& \multirow{2}{*}{$({\cblue 2},\ {\cred 2})$} \\
\cellcolor{lgrey} ${\cred \widehat{p}^{red}_1} = 0\ {\cblue \Rightarrow \widehat{p}^{blue}_2} = 1$
&
&
& \\
\hline
\cellcolor{lgrey} ${\cblue \widehat{p}^{blue}_1} = 0$; \ ${\cblue \widehat{p}^{blue}_2} = \big(${\bf if}
& \multirow{2}{*}{$({\cblue 1},\ {\cred 1})$}
& \multirow{2}{*}{$({\cblue 2},\ {\cred 2})$}
& \multirow{2}{*}{$\big({\cblue 3 - \frac{1}{2(n-1)}},\ {\cred 3 - \frac{1}{2(n-1)}}\big)$} \\
\cellcolor{lgrey} ${\cred \widehat{p}^{red}_1} = 0$ {\bf then} 1 {\bf else} ${\cblue p^{blue}_2}\big)$
&
&
& \\
\hline
\end{tabular}
}
\end{center}


\begin{Note}
If we replaced ${\cgreen -6}$ to, say, ${\cgreen -3}$, and therefore, the efficient decision is \textnormal{YES} if and only if {\bf either or both} agents have \textnormal{HIGH} type, then we would see the same issue.
\end{Note}

This is our simplest but not nearly our strongest example for instability.
We will see that the truthful strategy profile is even less stable by having more rounds or relaxing the idealized information structure.

This example is degenerate in multiple senses, but these are not the cause of instability. 
In Appendix~\ref{app:details}, we will show a more convincing but more difficult version of this example with a more complete analysis. 


\subsection{Explanation of the example and generalizations}

The balanced Team Mechanism means the following. In round $t$, {\cblue Blue} and {\cred Red} have to report about the current probability ${\cblue p^{blue}_t}$ and ${\cred p^{red}_t}$ (respectively) that his final type will be HIGH. These reports are denoted by ${\cblue \widehat{p}^{blue}_t}$ and ${\cred \widehat{p}^{red}_t}$, respectively. (In fact, they should report their types, but the mechanism will use only these probabilities.) Then we have the following transfers
\begin{alignat}{1}
{\cblue \gamma^{blue}_t} &= 
({\cred 4}\ {\cgreen -\ 6}) \cdot {\cred \widehat{p}^{red}_{t-1}} \cdot {\cblue \widehat{p}^{blue}_t}
- ({\cred 4}\ {\cgreen -\ 6}) \cdot {\cred \widehat{p}^{red}_{t-1}} \cdot {\cblue \widehat{p}^{blue}_{t-1}}
= 2 \cdot {\cred \widehat{p}^{red}_{t-1}} \cdot \big({\cblue \widehat{p}^{blue}_{t-1}} - {\cblue \widehat{p}^{blue}_t}\big) \label{eq:gb} \\
{\cred \gamma^{red}_t} &=
({\cblue 4}\ {\cgreen -\ 6}) \cdot {\cblue \widehat{p}^{blue}_{t-1}} \cdot {\cred \widehat{p}^{red}_t}
- ({\cblue 4}\ {\cgreen -\ 6}) \cdot {\cblue \widehat{p}^{blue}_{t-1}} \cdot {\cred \widehat{p}^{red}_{t-1}}
= 2 \cdot {\cblue \widehat{p}^{blue}_{t-1}} \cdot \big({\cred \widehat{p}^{red}_{t-1}} - {\cred \widehat{p}^{red}_t}\big) \label{eq:gr}
\end{alignat}
paid equally by the $n-1$ other agents.
This formula is coming from \eqref{gammadef} (or see in \cite{AtSe}, page 2473, equation (5)) in the following way. If both agents will have HIGH type, then the public decision will be YES, and the total utility of the agents excluding {\cblue Blue} or {\cred Red} will be $4 + (-6) = -2$. It will happen with probability ${\cblue p^{blue}} \cdot {\cred p^{red}}$.

From \eqref{eq:gb} and \eqref{eq:gr} and $\E(p^i_{t+1} \where p^i_t) = p^i_t$, the Reader can easily calculate the values in the table. But the intuition behind the issue is the following.

\eqref{eq:gb} can be interpreted as ${\cblue \widehat{p}^{blue}}$ is a quantity of a good that {\cblue Blue} can buy and sell for the changing unit price of $2 \cdot {\cred \widehat{p}^{red}_{t-1}}$ in round $t$.
As a consequence, if an agent expects this unit price to increase (or decrease), then he has an incentive to buy (or sell, respectively) as much as he can.
Specifically, the truthful strategy of {\cblue Blue} is weakly dominated by its modification that if ${\cred \widehat{p}^{red}_{k-1}} = 0$, then he always reports ${\cblue \widehat{p}^{blue}_{k}} = 1$.

First, we show an intuitive reason why the truthful strategy profile is not rational.
If in every odd round both agents report ${\cblue \widehat{p}^{blue}_{2t+1}} = {\cred \widehat{p}^{red}_{2t+1}} = 1$, and in every even round both report ${\cblue \widehat{p}^{blue}_{2t}} = {\cred \widehat{p}^{red}_{2t}} = 0$, then in every even round both get a transfer of $2 \cdot 1 \cdot (1-0) = 2$, and in every odd round both get a transfer $2 \cdot 0 \cdot (0-1) = 0$.
So this deviation from the truthful strategies is very beneficial for both agents.

This deviation does not require coordination: if either agent starts playing it, then the other agent is incentivized to join.
For example, if {\cblue Blue} reports probability ${\cblue \widehat{p}^{blue}_t} = 0$, then it means that the good is free for {\cred Red}, so {\cred Red} should buy as much as he can. Therefore, {\cblue Blue} expects his unit price to increase, as well, so he also buys as much as he can. Namely, both agents report probability ${\cblue \widehat{p}^{blue}_{t+1}} = {\cred \widehat{p}^{red}_{t+1}} = 1$.



We can modify the example so that the agents never report probability 0 before the last round. Namely, if the rules of the game include that, say, ${\cblue p^{blue}_{1}}, {\cred p^{red}_{1}} \ge 0.1$ possibly with equality, then the same argument holds with $\widehat{p}^i_{1} = 0.1$ instead of $\widehat{p}^i_{1} = 0$.

From a more general point of view,

 the fact that the truthful strategy profile is a PBE strongly relies on their belief that the reports of the other agents are the best

In more general, if the model allows an arbitrarily small probability that {\cblue Blue} will have a bit different expectation of {\cred $p^{red}_{t}$} from the latest report ${\cred \widehat{p}^{red}_{t-1}}$, then in every PBE, these two agents will report high and low probabilities alternately, and no PBE will remain close to the truthful strategy profile.

We note that if we replaced ${\cgreen 6}$ to ${\cgreen 3}$, then ${\cblue \widehat{p}^{blue}_{2t}} = 0$, ${\cred \widehat{p}^{red}_{2t}} = 1$, ${\cblue \widehat{p}^{blue}_{2t+1}} = 1$, ${\cred \widehat{p}^{red}_{2t+1}} = 0$ would be a beneficial deviation in the same sense. 



\section{Further weaknesses of the balanced Team Mechanism} \label{sec:weaknesses}

\subsection{Independent types, information} \label{sec:type-info}

The balanced Team Mechanism uses the assumption of \emph{independent types}.
One may fail to realize that this is a much stronger assumption in a dynamic stochastic game than, say, in an auction game.

First of all, notice that ``type'' can mean two different things. We call them ``payoff type'' and ``full type''. Full type is a synonym of information, but payoff type only includes the information of the player that directly affects his payoff. For example, Harsányi described Bayesian games by full types. But in a game with perfect information, the full types of the players do not really make sense because the full types are always the same for every player (except that each player knows his identity in the game). But payoff types may make sense here.
In an auction game, the type of a player typically means payoff type, namely, his own valuation of the good(s).

Independent payoff types is a much more reasonable assumption than independent full types.
As we show in Appendix~\ref{sec:type}, the Athey--Segal paper uses neither payoff type nor full type. The only understanding we found that makes the results formally correct is the following. Type means full type excluding the public information. Or with an equivalent interpretation, type means payoff type, but they had a hidden assumption that the information of each agent consists only of his payoff type and the public information.

Whichever interpretation we accept, this is still a very strong and practically unrealistic assumption. For example, this assumption implies that an agent cannot reveal anything about the history of his \emph{past types}.
We can see in Section~\ref{sec:wrong-ex} that the mechanism has no tolerance for any relaxation of this assumption. Appendix~\ref{sec:simple-example} can help further understand the nature of this issue.

\subsection{Dependence on the initial types}

The balanced Athey--Segal mechanism is dependent on the initial types of the agents $\theta^N_0$. This assumes that the initial types are publicly known and verifiable by a court.
We can interpret this assumption in different ways, but all of them look practically unreasonable, especially in contrast to the assumption of independent types.

For example, if we are interpreting $\theta^i_0$ as  common knowledge about the probability distribution of the true type $\theta^i_1$, then it means that the designer can exactly define the common prior and no other agent can possibly have any further knowledge about the prior distribution of the types of others.
Based on Example~\ref{general-example} (in Section~\ref{sec:wrong-ex}), we can see that if an agent has an arbitrary small correlated information about the types of others, then there exists no PBE close in any sense to the truthful strategy profile.

We note that the unbalanced Team Mechanism is not dependent of the initial types of the agents $\theta^N_0$, but this difference between the two mechanisms was not pointed out in the paper.


\section{How to fix the balanced Team Mechanism} \label{sec:fix}


To fix the Athey--Segal mechanism, we need to separate payoff type $\theta^i_t$ from information (full type) $\phi^i_t$ and $\phi^i_{t+}$, where $\phi^i_t$ is the information of the agent when he reports $\widehat{\theta}^i_t$, and $\phi^i_{t+}$ is the information of the agent just before the next-round state $\theta^i_{t+1}$ is chosen by nature.

We define {\bf momentarily private payoff types} as follows. Each payoff type $\theta^i_t$ is independent of the joint information (full type) of the other players in the same round $\phi^{N \setminus \{i\}}_t$ (and everything before round $t$), conditional on $x_{t-1}$ and $\theta^i_{t-1}$.
We can fix the mechanism for momentarily private payoff types.
(In fact, we only need an even weaker assumption that each agent $i$ can keep his payoff type $\theta^i_t$ momentarily secret if he wants to.)

In the balanced Team Mechanism, the designer calculates $\gamma^i_t$ as a marginal contribution of $\theta^i_t$ given the reported types of the agents in the previous round $\widehat{\theta}^N_{t-1}$.
But the same proof would work if it was calculated given the same-round reports of others
$(\widehat{\theta}^{N \setminus \{i\}}_t, \widehat{\theta}^i_{t-1})$.
Or we could use any combination of the two rules.

We recommend one combination, which is coming from the quasi-dominant equilibrium implementation of a more complex model in \cite{CsE} (a tendering model with arbitrary private initial types of the agents), specified for our case. Namely, we update the types one by one. Say, we calculate the change $\gamma^i_t$ given the reported types $(\widehat{\theta}^1_t, \widehat{\theta}^2_t, ...,$ $\widehat{\theta}^{i-1}_t,$ $\widehat{\theta}^i_{t-1}, \widehat{\theta}^{i+1}_{t-1}, ..., \widehat{\theta}^n_{t-1})$. Formally, \eqref{gammadef} is replaced by the following.
\begin{equation*} 
    \gamma^i_t = \sum_{j \in N \setminus\{i\}} \sum_{t' = t}^T \E \Big(
    \upsilon_{t'}^{j, \xi}(\widehat{\theta}^1_t, \widehat{\theta}^2_t, ..., \widehat{\theta}^{i-1}_t, \widehat{\theta}^i_t, \widehat{\theta}^{i+1}_{t-1}, ..., \widehat{\theta}^n_{t-1})
    - \upsilon_{t'}^{j, \xi}(\widehat{\theta}^1_t, \widehat{\theta}^2_t, ..., \widehat{\theta}^{i-1}_t, \widehat{\theta}^i_{t-1}, \widehat{\theta}^{i+1}_{t-1}, ..., \widehat{\theta}^n_{t-1}) \Big)
\end{equation*}

We also change the rule of how $\gamma^i_t$ is paid by the other agents. Instead of sharing it equally between the other agents, we say that whoever gets affected positively (or negatively) by a new report pays (or gets) the same transfer.
For example, if {\cblue Blue} makes a new report, and it changes the expected total utility of {\cred Red} by {\cred $\$15$}, and of {\cgreen Green} by $-\$5$, then ${\cblue Blue}$ gets ${\cred \$15} - {\cgreen \$5} = {\cblue \$10}$.
With the original transfer rule, each of the $n-1$ other agents pays $\frac{\cblue \$10}{n-1}$. With our new rule, {\cred Red} pays {\cred $\$15$} and {\cgreen Green} gets {\cgreen $\$5$}, and the others pay 0. Formally, the payment to $i$ is the following.
\begin{equation*} 
    y^i_t = \sum_{j \in N \setminus\{i\}} \sum_{t' = t}^T \E \Big(
    \upsilon_{t'}^{j, \xi}(\widehat{\theta}^1_t, \widehat{\theta}^2_t, ..., \widehat{\theta}^{i-1}_t, \widehat{\theta}^i_t, \widehat{\theta}^{i+1}_{t-1}, ..., \widehat{\theta}^n_{t-1})
    - \upsilon_{t'}^{j, \xi}(\widehat{\theta}^1_t, \widehat{\theta}^2_t, ..., \widehat{\theta}^{i-1}_t, \widehat{\theta}^i_{t-1}, \widehat{\theta}^{i+1}_{t-1}, ..., \widehat{\theta}^n_{t-1}) \Big)
\end{equation*}
\begin{equation*} 
    - \sum_{j \in N \setminus\{i\}} \sum_{t' = t}^T \E \Big(
    \upsilon_{t'}^{i, \xi}(\widehat{\theta}^1_t, \widehat{\theta}^2_t, ..., \widehat{\theta}^{j-1}_t, \widehat{\theta}^j_t, \widehat{\theta}^{j+1}_{t-1}, ..., \widehat{\theta}^n_{t-1})
    - \upsilon_{t'}^{i, \xi}(\widehat{\theta}^1_t, \widehat{\theta}^2_t, ..., \widehat{\theta}^{j-1}_t, \widehat{\theta}^j_{t-1}, \widehat{\theta}^{j+1}_{t-1}, ..., \widehat{\theta}^n_{t-1}) \Big)
\end{equation*}

Under this mechanism, the truthful strategy profile is an {\bf efficient Nash equilibrium with guaranteed expected payoffs}. It means the following (with the payoff functions $F_i$).
\begin{itemize}
    \item The truthful strategy $s^*_i$ of every agent $i$ guarantees an expected payoff of at least $C_i$ no matter which strategies the other agents choose. Namely, $\forall s_{N \setminus i}\: \E\big(F_i(s^*_i, s_{N\setminus i})\big) \ge C_i$;
    \item  These guaranteed expected payoffs $C_i$ sum up to the total expected payoffs of all agents with the efficient strategy profile. Namely, $\sum\limits_{i \in N} C_i = \sup\limits_{s_N} \sum\limits_{i \in N} \E\big(F_i(s_N)\big)$.
\end{itemize}
Notice that it implies collusion-resistance, as well.
All these are also true in a model with private decisions. \cite{CsE}

We note that for 2 players and no same-round chance events, the balanced Team Mechanism coincides with this new mechanism. This is the reason why the example in the Athey--Segal paper (Section~3 in \cite{AtSe}) works perfectly well.

For the (degenerate) case of same-time chance events, the mechanism depends on the ordering of the indexes of the agents. But any convex combination of these mechanisms satisfy the same properties. Therefore, if we want a symmetric mechanism for the agents, then we can average the payment rules for all permutations of the agents. In other words, $\gamma^i_t$ is the Shapley contribution of the report $\widehat{\theta}^i_t$ among the reports $\widehat{\theta}^N_t$ to the change in the expected total utility of others (by trustful calculation). For Example~\ref{general-example}, this means
\begin{alignat*}{2}
{\cblue \gamma^{blue}_t} &&= 100 \cdot \frac{{\cred \widehat{p}^{red}_{t-1}} + {\cred \widehat{p}^{red}_t}}{2} \big({\cblue \widehat{p}^{blue}_{t-1}} - {\cblue \widehat{p}^{blue}_t}\big), \\
{\cred \gamma^{red}_t} &&= 100 \cdot \frac{{\cblue \widehat{p}^{blue}_{t-1}} + {\cblue \widehat{p}^{blue}_t}}{2} \big({\cred \widehat{p}^{red}_{t-1}} - {\cred \widehat{p}^{red}_t}\big).
\end{alignat*}
If we use this new mechanism, then it resolves the problems shown about Example~\ref{general-example}. If the two agents play ${\cblue \widehat{p}^{blue}_{2t+1}} = {\cred \widehat{p}^{red}_{2t+1}} = 1$ and ${\cblue \widehat{p}^{blue}_{2t}} = {\cred \widehat{p}^{red}_{2t}} = 0$, then they buy and sell the good for always the same unit price of $100 \cdot \frac{0+1}{2} = 100 \cdot \frac{1+0}{2} = 50$.

Note that if we use a continuous-time model instead of the round-by-round model, then for the calculation of $\gamma^i_t$, we can use the order of the receiving times of the reports. And hereby same-time reports do not normally happen. 
In a continuous-time model, momentarily private payoff types only mean that for each update of a payoff type of an agent, this agent can be the first to report about it, and nobody else can observe and report any correlated information any faster.


We have only one weakness that we did not resolve: the initial types are still fixed.
This issue cannot be fixed as nicely as the other problems, but this is possible to handle quite well, due to the property of guaranteed expected payoffs.
Namely, this mechanism gives an ``almost reduction'' of the dynamic stochastic problem to a single-round problem.
For example in \cite{CsE}, we are discussing a situation where a principal wants to choose some of the competing agents for a dynamic stochastic multi-agent working process, with no prior assumption about the types of the agents.

\section{The more general model and the stronger results} \label{sec:myresult}

We show the model and mechanism in Section~\ref{sec:fix} not as a comparison to the Athey--Segal paper, but purely alone. Namely, we show the more general model but still with fixed initial types, and how the direct mechanism in \cite{CsE} applies here.
In Appendix~\ref{app:non-qlinear}, we will show how the indirect mechanism in \cite{CsE} can be applied for this setup, and how it extends nicely to the case of nonquasilinear payoffs.

\subsection{General notation}

Each agent $i$ at each time point $m$ has a {\bf payoff type} $\theta^i_m \in \Theta$ and {\bf information} (or full type) $\phi_m^i \in \mathcal{I}$.
{\bf Strategy} is a mapping from information to actions.
We always assume that the information $\phi_m^i \in \mathcal{I}$ of each agent $i$ always includes his earlier information $\phi_{m-1}^i \in \mathcal{I}$, his payoff type $\theta^i_m$ and the history of ``public'' actions (or decisions).\footnote{Formally, for example, $\phi_m^i$ includes $\theta^i_m$ means that there exists a public function $\tau: \mathcal{I} \to \Theta$ such that if agent $i$ has information $\phi_m^i$, then his payoff type must be $\theta^i_m = \tau(\phi_m^i)$.}

In contrast to the standards in theoretical economics, if we do not specify something in the model, then it will mean that it is not specified. In other words, any setup which satisfies the specifications belong to the set of setups we are discussing.

\subsection{The extended model} \label{simple-sec}

We have a set $N = \{1, 2, ..., n\}$ of agents. There is a finite number of rounds $T = \{1, 2, \dots k\}$, each round consists of 4 steps (subrounds). $\Theta$ denotes the finite set of possible {\bf payoff types}. The initial payoff types $\theta^{N_0}_0 \in \Theta^{N_0}$ are fixed, where $N_0 = N \cup \{0\}$ and $\theta^0_t$ is a public type in round $t$.

\smallskip
\noindent
Each round $t \in T$ consists of the following steps.
\begin{itemize}
    \item The planner makes a public decision $x^0_t \in X$.
    \item Each agent $i \in N$ makes a decision $x^i_t \in X$.
    \item For every agent $i \in N_0$, Nature chooses $\theta^i_t \in \Theta$ from a probability distribution $\mu(\theta^i_{t-1}, \theta^0_{t-1}, x^0_t, x^i_t)$ (where $\mu \: \Theta^2 \times X^2 \rightarrow \Delta(\Theta)$ is a given public function) independently from each other and from the state of the game,\footnote{As a further technical assumption, if nature makes any move between this step and the next step in round $t$, then its distribution must be independent of $\theta^N_t$.} conditional on these probability distributions.
    \item Each agent $i \in N$ sends a report $\widehat{\theta}^i_t$.
\end{itemize}

There might be further actions by the agents or nature (e.g.\ cheap talk, signaling, type revelation) that may affect the information of the agents (but not their payoff types).

\smallskip
The information of the planner includes the initial payoff types $\theta^{N_0}_0$ and the history of reports and public states $\widehat{\theta}^{N_0}_T$ (up to the current time point).
At the end, the planner determines transfers $y^i$ to agent $i \in N$ with $\sum\limits_{i \in N} y^i = 0$.

The utility of each agent $i \in N$ is $u_i = v(\theta^i_{k}) + y^i$ (for a given $v \: \Theta \rightarrow \R$).

\medskip
{\bf  Assumption of momentarily private types.}
Denote by $\phi^i_t$ the information of agent $i$ when he reports $\widehat{\theta}^i_t$.
For every $i \in N$, $\theta^i_t$ must be independent of $\phi^{N-i}_t$ (and everything before round $t$) conditional on $\theta^i_{t-1}$.
(\emph{Weaker assumption.} If $i$ is using the truthful strategy, then $\theta^i_t$ must be independent of $(\phi^{N-i}_t, \phi^i_{t-1})$ conditional on $\theta^i_{t-1}$.)\footnote{We need to assume that the agent is able to keep his chance event his private information until he reports it. But the point of the weaker version is that we do not need to assume that he is not able to share this information.}

We note that in a continuous-time setup, this assumption only means that if an agent reports every chance event (changes in his type) as soon as he can, then the other agents cannot get to know any stochastic information which is not already reported.

\subsection{The mechanism} \label{sec:my-mechanism}

We show a direct mechanism, namely, we always ask the agents to report their types, and we make recommendations for their private decisions.
For calculations, the designer will be trustful in the sense that expected utilities will be calculated as if the agents had the same types as they reported, they would always report the truth and they would always make the hidden decision as recommended.
We fix an efficient decision policy $\xi: \{1, 2, ..., T\} \times \Theta^{N_0} \to X \times X^N$.
In round $t$, we learn the new types one by one, so $\widehat{\theta}_{t,i}^N = (\widehat{\theta}^1_t, \widehat{\theta}^2_t, ..., \widehat{\theta}^{j-1}_t, \widehat{\theta}^j_t, \widehat{\theta}^{j+1}_{t-1}, ..., \widehat{\theta}^n_{t-1})$ is our information about the reported types in round $t$ after learning the types of agents 1, 2, $...,\ j$.
Let $\Upsilon^i_{t, j}(\theta^N)$ denote the expected total utility of agent $i$ if $\widehat{\theta}_{t,j}^N = \theta^N$ (and with the trustful assumption). Then for each round $t \in \{1, 2, ..., T\}$ and each pair of different agents $(i, j) \in N \times N$, $j$ pays to $i$ a signed transfer of $\Upsilon^j_{t, i}(\theta^N) - \Upsilon^j_{t, i-1}(\theta^N)$.

\begin{Note}
We could ask the agents to report their types at the same time and average the evaluations with every permutation of the agents.
But if we use a (more realistic) continuous-time model, then the reports typically arrive one by one anyway.

In Appendix~\ref{app:non-qlinear}, we will sketch the more general non-revelation mechanism in \cite{CsE} which works importantly differently if the payoffs of the agents are not quasilinear.

\end{Note}

\subsection{The equilibrium concept and the proof}

As $u$ denoted utility in the other sections, in order to avoid ambiguity, let $F$ denote payoff.

\begin{Def}
In a stochastic dynamic game with a set of players $N$, suppose that there is a strategy profile $s_*^N \in \Str^N$ and constants $C^i$ satisfying the following.
\begin{align}
\forall i \in N,\  \forall s^{N-i} \in \Str^{N-i}\:&&
\E\big(F^i(s_*^i, s^{N-i})\big) &\ge C^i \label{pi2}
&&\phantom{\forall i \in N,\  \forall s^{N-i} \in \Str_{N-i}\:}
\\ &&
\sup_{s^N \in \Str^N} \sum_{i \in N} \E\big(F^i(s^N)\big) &= \sum_{i \in N} C^i \label{pN2}
&&\phantom{\forall s^N \in \Str^N\:}
\end{align}
Then $s_*^N$ is an \textbf{efficient Nash equilibrium with guaranteed expected payoffs}.
\end{Def}

\begin{proof}[Justification.]
\eqref{pi2} means that $i$ can guarantee himself an expected payoff $C^i$ by playing $s_*^i$. Each player $i \in N$ has no hope of getting more expected payoff than the maximum possible total expected payoff of all players minus the sum of the guaranteed expected payoffs of the other players. Therefore, $i$ has no hope of getting more expected payoff than
\begin{equation*}
\sup_{s^N \in \Str^N} \sum_{j \in N} \E\big(F^{j}(s_N)\big) - \sum_{j \in N \setminus \{i\}} \E\big(F^j(s_*^N)\big) \mathop{\le}^{\eqref{pN2}} \sum_{j \in N}
C^j - \sum_{j \in N \setminus \{i\}} C^j = C^i.
\end{equation*}
Consequently, each player $i \in N$ has no incentive to deviate from $s_*^i$, and therefore, we can rightfully say that $s_*^N$ is an equilibrium.

The same argument holds for coalitions: for any $X \subset N$, no joint strategy profile can provide them a higher total expected payoff than
\begin{equation*}
\sup_{s^N \in \Str^N} \sum_{j \in N} \E\big(F^{j}(s^N)\big) - \sum_{j \in N \setminus X} \E\big(F^j(s_*^N)\big) \mathop{\le}^{\eqref{pN2}} \sum_{j \in N}
C^j - \sum_{j \in N \setminus X} C^j = \sum_{i \in X} C^i. \qedhere
\end{equation*}
\end{proof}

\begin{Theorem}
The truthful strategy profile is an efficient Nash equilibrium with guaranteed expected payoffs, with the truthful strategy profile $s_*^N$ and $C^i = \Upsilon^i_{0, 0}$.
\end{Theorem}

\begin{proof}

For \eqref{pi2}, it is enough to prove that $\Upsilon^i$ plus the total transfers is a martingale if $i$ is truthful because this sum is $C^i$ at the beginning and $u_i(s_*^i, s^{N-i})$ at the end.
The martingale property holds when the agent receives his next-round type or the public type is updated because of the martingale property of the expected value. And whenever the other agent $j$ changes $\Upsilon^i$, he pays to $i$ the signed difference, so the sum is invariant here.

\eqref{pN2} holds because the truthful strategy profile maximizes the expected total payoff.
\end{proof}

\section{Acknowledgement}

I would like to thank Ilya Segal, Larry Samuelson, Johannes Hörner, and Alexander Rodivilov for their help in better understanding and presenting the results.

\bibliography{AS-wrong}

\newpage
\begin{center}
{\huge Appendix}
\end{center}
\bigskip
\appendix

\section{A more generic counterexample} \label{app:details}

We present a longer but more convincing example to show that the issues with the Athey--Segal balanced Team Mechanism were not caused by the degenerations of our previous examples.

In this section, we will specify and slightly modify Example~\ref{general-example}.
Our analysis will be based on Lemma~\ref{thm:parity}.
In order to get convenient values, we replace the utilities of $(\textnormal{\cblue Blue}, \textnormal{\cred Red}, \textnormal{\cgreen Green})$ to $({\cblue 84\textnormal{ or }104}, {\cred 84\textnormal{ or }104}, {\cgreen -204})$, and hereby the constant factor $6 - 4 = 2$ is replaced by $204 - 104 = 100$ in \eqref{eq:gb} and \eqref{eq:gr}.
We use the notation $\delta_t^i = \widehat{p}^i_t - p^i_t$.

\begin{Lemma} \label{thm:parity}
Assume that ${\cblue \widehat{p}^{blue}_t} = {\cblue p^{blue}_t}$ in every even round $t$, and ${\cred \widehat{p}^{red}_t} = {\cred p^{red}_t}$ in every odd round $t$, and both are true in the first and last rounds $t = 0$ and $t = k$. Then the followings hold.
\begin{equation}
    \E ( {\cblue \gamma^{blue}} )
    = \E \Big( \sum_{t=1}^k {\cblue \gamma^{blue}_t} \Big)
    = 100 \cdot \sum_{t=1}^k \E \big(-{\cred \delta^{red}_{t-1}} \cdot {\cblue \delta^{blue}_{t}} \big)
    = 100 \cdot \sum_{t=1}^{ \lfloor \frac{k-2}{2} \rfloor} \E \big(-{\cred \delta^{red}_{2t}} \cdot {\cblue \delta^{blue}_{2t+1}} \big) \label{eq:Egb}
\end{equation}
\begin{equation}
    \E ( {\cred \gamma^{red}} )
    = \E \Big(\sum_{t=1}^k {\cred \gamma^{red}_t} \Big)
    = 100 \cdot \sum_{t=1}^k \E \big(-{\cblue \delta^{blue}_{t-1}} \cdot {\cred \delta^{red}_{t}}\big)
    = 100 \cdot \sum_{t=1}^{\lfloor \frac{k-1}{2} \rfloor} \E \big(-{\cblue \delta^{blue}_{2t-1}} \cdot {\cred \delta^{red}_{2t}}\big)
    \label{eq:Egr}
\end{equation}
\end{Lemma}

\medskip

\subsection{The setup}

Consider now Example~\ref{general-example} with $k = 4$ number of rounds (and $n \ge 3$ players). We define a finite space of types for both players, the types are encoded by
\begin{center}
[{\cblue b}/{\cred r}: {\cblue Blue} or {\cred Red}][Round number]:[the percentage that his type will be HIGH]\%.
\end{center}
\begin{minipage}{\textwidth}
\begin{multicols}{2}
\phantom{000000000}\emph{Types of agent {\cblue Blue}:} 
\\
\\ \phantom{00000000000000} {\cblue b0:50\%}
\\ \phantom{000000000} {\cblue b1:30\%}, \phantom{00} {\cblue b1:70\%}
\\ \phantom{000000} {\bf{\cblue b2:20\%}, \phantom{000000} {\cblue b2:80\%}}
\\ \phantom{000} {\cblue b3:10\%}, \phantom{00000000000000} {\cblue b3:90\%}
\\ \phantom{} {\bf{\cblue b4:0\%}, \phantom{00000000000000000} {\cblue b4:100\%}}
\\ ({\cblue LOW}) \phantom{000000000000000000000} ({\cblue HIGH})
\columnbreak

\phantom{00000000}\emph{Types of agent {\cred Red}:}
\\
\\ \phantom{00000000000000} {\cred r0:50\%}
\\ \phantom{00000000000000} {\cred r1:50\%}
\\ \phantom{000000} {\cred r2:20\%}, \phantom{00000000} {\cred r2:80\%}
\\ \phantom{000} {\bf{\cred r3:10\%}, \phantom{00000000000} {\cred r3:90\%}}
\\ \phantom{} {\bf{\cred r4:0\%}, \phantom{00000000000000000} {\cred r4:100\%}}
\\ ({\cred LOW}) \phantom{00000000000000000000.} ({\cred HIGH})
\end{multicols}
\end{minipage}

\bigskip

The transition probabilities can be calculated from (the martingale property of) the percentages. E.g., {\cblue b1:30\%} is transitioning to {\cblue b2:20\%} or to {\cblue b2:80\%} with probabilities $5/6$ and $1/6$, respectively, because ${\cblue 30\%} = \frac{5}{6} \cdot {\cblue 20\%} + \frac{1}{6} \cdot {\cblue 80\%}$.


We make some modifications to the setup in order to incentivize the agents (under the balanced Team Mechanism) to tell the truth in specific rounds (marked with {\bf bold}). Hereby, we will be able to focus on the possibilities of deviations only in the rest of the rounds, which will simplify the analysis.

We add 4 extra public decisions, each of them affecting only one agent. In round 2, there is a public decision "{\cblue b2:20\%}" or "{\cblue b2:80\%}". If this decision does not coincide with the true type of {\cblue Blue}, then {\cblue Blue} gets utility $-10^{42}$ (instead of $0$) for that round. We do the analogous modification in rounds {\cblue 2} and {\cblue 4} for {\cblue Blue} and in rounds {\cred 3} and {\cred 4} for {\cred Red} (marked with {\bf bold}).  (Formally, $X_2 = \{ "\text{\cblue b2:20\%}", "\text{\cblue b2:80\%}" \}$, $X_3 = \{ "\text{\cred r3:10\%}", "\text{\cred r3:90\%}" \}$, $X_4 = \{ "\text{\cblue b4:0\%}", "\text{\cblue b4:100\%}" \} \times \{ "\text{\cred r4:0\%}", "\text{\cred r4:100\%}" \} \times \{ "\text{YES}", "\text{NO}" \}$.)

\subsection{The analysis}

Under the balanced Team Mechanism, these modifications in the setup induce no transfers between the agents. The only effect is that the agents get a huge punishment if they do not report truthfully in the specified rounds.
Therefore, by (a dynamic version of) \emph{strict domination}, we eliminate the possibilities of not telling the truth about these 4 extra public decisions (marked with {\bf bold}). It leaves a total of 3 binary decisions for the two players: ${\cblue \widehat{p}^{blue}_1}$, ${\cred \widehat{p}^{red}_2}$ and ${\cblue \widehat{p}^{blue}_3}$.
The agents only observe their own types and these earlier decisions of the other agent.

In this reduced game, the utilities $u(x_t, \theta^i_t)$ are unaffected by these three decisions, and the conditions of Lemma~\ref{thm:parity} apply. Therefore,
\begin{itemize}
    \item {\cblue Blue} is maximizing
    $\E \big( {\cblue \gamma^{blue}}\big) - \frac{1}{n-1} \E \big({\cred \gamma^{red}}\big) = 100 \cdot \E \big(-{\cred \delta^{red}_2} \cdot {\cblue \delta^{blue}_3}\big) - \frac{100}{n-1} \E \big(-{\cblue \delta^{blue}_1} \cdot {\cred \delta^{red}_2}\big)$,
    \item {\cred Red} is maximizing $\E \big({\cred \gamma^{red}}\big) - \frac{1}{n-1} \E \big( {\cblue \gamma^{blue}}\big) = 100 \cdot \E \big(-{\cblue \delta^{blue}_1} \cdot {\cred \delta^{red}_2}\big) - \frac{100}{n-1} \E \big(-{\cred \delta^{red}_2} \cdot {\cblue \delta^{blue}_3} \big)$.
\end{itemize}

We refer to the binary options and types by ``low'' and ``high'', and in this sense, we can say that two decisions are of the ``same kind'', denoted by "$\sim$", or ``opposite'', denoted by "$\nsim$". Notice that the reduced game is symmetric to low and high.

Let us start with the last decision ${\cblue \widehat{p}^{blue}_3}$. As \eqref{eq:Egb} and \eqref{eq:Egr} show, it can only affect $\E \big( {\cblue \gamma^{blue}}\big) = 100 \cdot \E(-{\cred \delta^{red}_2} \cdot {\cblue \delta^{blue}_3})$. If ${\cred \widehat{p}^{red}_2} = {\cred 20\%}$, then $-{\cred \delta^{red}_2} \ge 0$, therefore, choosing ${\cblue \widehat{p}^{blue}_3} = {\cblue 90\%}$ weakly dominates ${\cblue \widehat{p}^{blue}_3} = {\cblue 10\%}$.
Analogously, if ${\cred \widehat{p}^{red}_2} = {\cred 80\%}$, then $-{\cred \delta^{red}_2} \le 0$, therefore, choosing ${\cblue \widehat{p}^{blue}_3} = {\cblue 10\%}$ weakly dominates ${\cblue \widehat{p}^{blue}_3} = {\cblue 90\%}$.
Hereby we could conclude by weak dominance that {\cblue Blue} should choose ${\cblue \widehat{p}^{blue}_3} \nsim {\cred \widehat{p}^{red}_2}$.

If we want to be more careful with elimination by \emph{weak dominance}, then we can use instead the following argument.
${\cblue \widehat{p}^{blue}_3} \nsim {\cred \widehat{p}^{red}_2}$ is strictly better for {\cblue Blue} unless if ${\cred \widehat{p}^{red}_2} = {\cred p^{red}_2}$.
Therefore, in every Nash equilibrium,
\begin{equation} \label{eq:p3}
    \P\big( {\cblue \widehat{p}_3^{blue}} \sim {\cred \widehat{p}^{red}_2} \nsim {\cred p^{red}_2} \big) = 0,
\end{equation}
or in other words, ${\cblue \widehat{p}_3^{blue}} \nsim {\cred \widehat{p}^{red}_2}$ whenever ${\cblue \widehat{p}_3^{blue}}$ is relevant.

From now on, the formula ${\cblue \gamma^{blue}}$ will be calculated as a function of ${\cblue p^{blue}_1}$, ${\cblue \widehat{p}^{blue}_1}$ and ${\cred \widehat{p}^{red}_2}$ and with the assumption that ${\cblue \widehat{p}_3^{blue}} \nsim {\cred \widehat{p}^{red}_2}$. But we will keep in mind that {\cblue Blue} will have a final decision with the only effect that it can possibly decrease ${\cblue \gamma^{blue}}$.

\smallskip

It leaves a total of 2 binary decisions for the two players: ${\cblue \widehat{p}^{blue}_1}$ and ${\cred \widehat{p}^{red}_2}$.
Consider now the second decision ${\cred \widehat{p}^{red}_2}$. It has an effect on $\E \big({\cred \gamma^{red}}\big) = 100 \cdot \E(-{\cblue \delta^{blue}_1} \cdot {\cred \delta^{red}_2})$ and $\E \big( {\cblue \gamma^{blue}}\big) = 100 \cdot \E(-{\cred \delta^{red}_2} \cdot {\cblue \delta^{blue}_3})$.
\begin{itemize}
    \item The effect on $\E \big({\cred \gamma^{red}}\big)$. If ${\cblue \widehat{p}^{blue}_1} \nsim {\cblue p^{blue}_1}$, then ${\cred \widehat{p}^{red}_2} \nsim {\cblue \widehat{p}^{blue}_1}$ makes
    $\E \big({\cred \gamma^{red}}\big)$ higher by $100 \cdot ({\cblue 0.7} - {\cblue 0.3}) \cdot ({\cred 0.8} - {\cred 0.2}) = 24$. If ${\cblue \widehat{p}^{blue}_1} = {\cblue p^{blue}_1}$, then $\E \big({\cred \gamma^{red}}\big) = 0$ independently of ${\cred \widehat{p}^{red}_2}$.
    \item The effect on $\E \big( {\cblue \gamma^{blue}}\big)$. If ${\cblue p^{blue}_3} \sim {\cred p^{red}_2}$, then ${\cred \widehat{p}^{red}_2}$ has no effect on $\E \big( {\cblue \gamma^{blue}}\big)$.
    But if ${\cblue p^{blue}_3} \nsim {\cred p^{red}_2}$, then ${\cred \widehat{p}^{red}_2 \nsim p^{red}_2}$ increases $\E \big( {\cblue \gamma^{blue}}\big)$ by $100 \cdot ({\cred 0.8} - {\cred 0.2}) \cdot ({\cblue 0.9} - {\cblue 0.1}) = 48$.
\end{itemize}
It implies that {\cred Red} should choose 
${\cred \widehat{p}^{red}_2} \nsim {\cblue \widehat{p}^{blue}_1}$ or ${\cred \widehat{p}^{red}_2} = {\cred p^{red}_2}$.
Namely, if both decisions would be the same, then {\cred Red} is strictly better by choosing it. Otherwise, the best choice depends on his belief about the probability that ${\cblue \widehat{p}^{blue}_1} \nsim {\cblue p^{blue}_1}$.

More formally, assume by contradiction that in any Nash equilibrium, $\eps = \P\big({\cred \widehat{p}^{red}_2} \sim {\cblue \widehat{p}^{blue}_1} \nsim {\cred p^{red}_2} \big) > 0$.
\eqref{eq:p3} implies that 
$0 < \eps = \P\big({\cred \widehat{p}^{red}_2} \sim {\cblue \widehat{p}^{blue}_1} \nsim {\cred p^{red}_2} \big)
= \P\big({\cred \widehat{p}^{red}_2} \sim {\cblue \widehat{p}^{blue}_1} \nsim {\cred p^{red}_2} \sim {\cblue \widehat{p}^{blue}_3} \big)$.
If {\cred Red} chose ${\cred \widehat{p}^{red}_2} = {\cred p^{red}_2}$ instead, then it would weakly increase $\E \big({\cred \gamma^{red}}\big)$, and strictly decrease $\E \big( {\cblue \gamma^{blue}}\big)$ by $48 \cdot \eps$ independently of ${\cblue \widehat{p}^{blue}_3}$. Therefore, this deviation would strictly increase the payoff of {\cred Red}, which is a contradiction.

\smallskip

Now the Reader can jump to the matrix games, but we explain in short the dilemma about the first decision ${\cblue \widehat{p}^{blue}_1}$.
It can affect both $\E({\cred \gamma^{red}})$ and $\E({\cblue \gamma^{blue}})$.
If {\cred Red} chooses the strategy ${\cred \widehat{p}^{red}_2} = {\cred p^{red}_2}$, then it does not matter what {\cblue Blue} does.
So consider the case when {\cred Red} uses his other strategy ${\cred \widehat{p}^{red}_2} \nsim {\cblue \widehat{p}^{blue}_1}$.
In this case, $\E({\cblue \gamma^{blue}})$ is $48$ or $0$, and the probabilities depend on ${\cblue \widehat{p}^{blue}_1}$, and also $\P\big( {\cred \delta^{red}_2} \ne 0 \big) = 1/2$ independently of ${\cblue \widehat{p}^{blue}_1}$.
But for example, if ${\cblue p^{blue}_1} = {\cblue 30\%}$, then they will choose ${\cblue \widehat{p}^{blue}_1} = {\cblue 70\%}$, then ${\cred \widehat{p}^{red}_2} = {\cred 20\%}$, and then ${\cblue \widehat{p}^{blue}_3} = {\cblue 90\%}$.
Therefore, 
\begin{equation*}
    \P\big({\cblue \gamma^{blue}} = 48\big)
    = \P\big( {\cblue \delta^{blue}_3 \ne 0} \big) = \P\big({\cblue p^{blue}_3} = {\cblue 10\%} \bigwhere {\cblue p^{blue}_1} = {\cblue 30\%} \big) = 3/4.
\end{equation*}
With ${\cblue \widehat{p}^{blue}_1} = {\cblue p^{blue}_1} = {\cblue 30\%}$, it would be $ P\big({\cblue \gamma^{blue}} > 0\big) = \P\big({\cblue p^{blue}_3} = {\cblue 90\%} \bigwhere {\cblue p^{blue}_1} = {\cblue 30\%} \big) = 1/4 $.
This shows that {\cblue Blue} should report ${\cblue \widehat{p}^{blue}_1} = 1 - {\cblue p^{blue}_1}$, and therefore, {\cred Red} should report the opposite.

Assuming that {\cblue Blue} and {\cred Red} use a symmetric strategy for high and low, we get the following $2 \times 2$ matrix game. (It does not include that {\cblue Blue} is able to decrease ${\cblue \gamma^{blue}}$ by his last move.
We use a normalization factor of $1/3$ for $\big({\cblue \E(\gamma^{blue})},\ {\cred \E(\gamma^{red})}\big)$ in order to have smaller integers.)

\begin{center}
{\renewcommand{\arraystretch}{1.6}
\begin{tabular}[center]{c||c|c|}

$\frac{1}{3}\big({\cblue \E(\gamma^{blue})},\ {\cred \E(\gamma^{red})}\big)$ & \cellcolor{lgrey} ${\cred \widehat{p}^{red}_2} = {\cred p^{red}_2}$ & \cellcolor{lgrey} ${\cred \widehat{p}^{red}_2} \nsim {\cblue \widehat{p}^{blue}_1}$ \\
\hline
\hline
\cellcolor{lgrey} ${\cblue \widehat{p}^{blue}_1} = {\cblue p^{blue}_1}$ &  \cellcolor{lorange}$({\cblue 0},\ {\cred 0})$ & $({\cblue 2},\ {\cred 0})$ \\
\hline
\cellcolor{lgrey} ${\cblue \widehat{p}^{blue}_1} = 1 - {\cblue p^{blue}_1}$ & $({\cblue 0},\ {\cred 0})$ & $({\cblue 6},\ {\cred 4})$ \\
\hline
\end{tabular}
}
\end{center}

\noindent
We can see that {\cblue Blue} prefers ${\cblue \widehat{p}^{blue}_1} = 1 - {\cblue p^{blue}_1}$, and therefore, {\cred Red} should choose ${\cred \widehat{p}^{red}_2} \nsim {\cblue \widehat{p}^{blue}_1}$.
We can conclude it by iterative elimination of dominated strategies, or we could just use our intuitive understanding of rationality.

\smallskip

We could find some arguments to exclude the rationality of asymmetric strategies with respect to high and low, but it is easier to extend the matrix with all asymmetric strategies. The further pure strategies of {\cblue Blue} are always reporting ${\cblue p^{blue}_1} = {\cblue 30\%}$ and always reporting ${\cblue p^{blue}_1} = {\cblue 70\%}$. As for {\cred Red}, we only need to consider the options when ${\cblue \widehat{p}^{blue}_1} \sim {\cred p^{red}_2}$ because otherwise he should choose ${\cred \widehat{p}^{red}_2} = {\cred p^{red}_2}$. Therefore, the two extra pure strategies of red are the followings.
\begin{itemize}
    \item "preferably {\cred high}", meaning that ${\cred \widehat{p}^{red}_2} = {\cred 80\%}$ unless if $\big({\cblue \widehat{p}^{blue}_1} = {\cblue 70\%}$ and ${\cred p^{red}_2} = {\cred 20\%}\big)$;
    \item "preferably {\cred low}", meaning that ${\cred \widehat{p}^{red}_2} = {\cred 20\%}$ unless if $\big({\cblue \widehat{p}^{blue}_1} = {\cblue 30\%}$ and ${\cred p^{red}_2} = {\cred 80\%}\big)$.
\end{itemize}

Now we get the following $4 \times 4$ game. Remember that we got it by iterative elimination of \emph{only strictly dominated strategies} (in a dynamic sense) and then assuming that the last move is ${\cblue \widehat{p}^{blue}_3} \nsim {\cred \widehat{p}^{red}_2}$ justified by that ${\cblue \widehat{p}^{blue}_3} \sim {\cred \widehat{p}^{red}_2}$ never increases ${\cblue \E(\gamma^{blue})}$ and never changes ${\cred \E(\gamma^{red})}$.

\begin{center}
{\renewcommand{\arraystretch}{1.6}
\begin{tabular}[center]{c||c|c|c|c|}
$\frac{1}{3}\big({\cblue \E(\gamma^{blue})},\ {\cred \E(\gamma^{red})}\big)$ & \cellcolor{lgrey} ${\cred \widehat{p}^{red}_2} = {\cred p^{red}_2}$ & \cellcolor{lgrey} ${\cred \widehat{p}^{red}_2} \nsim {\cblue \widehat{p}^{blue}_1}$ & \cellcolor{lgrey} ${\cred \widehat{p}^{red}_2}$ pref. {\cred high} & \cellcolor{lgrey} ${\cred \widehat{p}^{red}_2}$ pref. {\cred low} \\
\hline
\hline
\cellcolor{lgrey} ${\cblue \widehat{p}^{blue}_1} = {\cblue p^{blue}_1}$ &  \cellcolor{lorange}$({\cblue 0},\ {\cred 0})$ & $({\cblue 2},\ {\cred 0})$ & $({\cblue 1},\ {\cred 0})$ & $({\cblue 1},\ {\cred 0})$ \\
\hline
\cellcolor{lgrey} ${\cblue \widehat{p}^{blue}_1} = 1 - {\cblue p^{blue}_1}$ & $({\cblue 0},\ {\cred 0})$ & $({\cblue 6},\ {\cred 4})$ & $({\cblue 3},\ {\cred 2})$ & $({\cblue 3},\ {\cred 2})$ \\
\hline
\cellcolor{lgrey} ${\cblue \widehat{p}^{blue}_1} = {\cblue 70\%}$ & $({\cblue 0},\ {\cred 0})$ & $({\cblue 4},\ {\cred 2})$ & $({\cblue 1},\ {\cred 0})$ & $({\cblue 3},\ {\cred 2})$ \\
\hline
\cellcolor{lgrey} ${\cblue \widehat{p}^{blue}_1} = {\cblue 30\%}$ & $({\cblue 0},\ {\cred 0})$ & $({\cblue 4},\ {\cred 2})$ & $({\cblue 3},\ {\cred 2})$ & $({\cblue 1},\ {\cred 0})$ \\
\hline
\end{tabular}
}
\end{center}



Notice that the truthful strategy profile is also a Nash equilibrium and a PBE: ${\cblue \widehat{p}^{blue}_1} = {\cblue p^{blue}_1}$ and ${\cred \widehat{p}^{red}_2} = {\cred p^{red}_2}$, and {\cblue Blue} decreases ${\cblue \E(\gamma^{blue})}$ to constant {\cblue 0} by ${\cblue \widehat{p}^{blue}_3} = {\cblue p^{blue}_3}$.

\subsection{Further important variants of the setup}

\begin{enumerate}
    \item In order to get a counterexample not only for truthfulness but also for efficiency, we can add a public decision to the setup in round 2, where {\cred Red} gets an additional utility ${\cred 1}$ if he reported the truth.
    It changes the matrix game a bit but otherwise it does not change the entire argument.
    \item If we modify the setup so that ${\cblue p^{blue}_1}$ is uniform random from the interval ${\cblue [30\%, 70\%]}$, then it makes the truthful strategy profile even less reasonable. Because if {\cblue Blue} reports ${\cblue \widehat{p}^{blue}_1} = {\cblue 30\%}$ or ${\cblue \widehat{p}^{blue}_1} = {\cblue 70\%}$, then it is even harder for {\cred Red} having no doubt that {\cblue Blue} was just telling the truth.
    \item If {\cblue Blue} has an option to reveal his type to {\cred Red} (which is a relaxation of the model), then the truthful strategy profile will no longer be subgame-perfect. Because it is easy to see that revealing ${\cblue \widehat{p}^{blue}_1} = 1 - {\cblue p^{blue}_1}$ would strictly incentivize {\cred Red} to choose ${\cred \widehat{p}^{red}_2} \nsim {\cblue \widehat{p}^{blue}_1}$, which is strictly better for {\cblue Blue}.
    (If {\cblue Blue} reveals his type to {\cred Red}, then a subgame starts from round 2. See Appendix~\ref{sec:belief} about the relation between PBE and subgame-perfection.)
\end{enumerate}

\subsection{Proof of Lemma \ref{thm:parity}} \label{app:Lemmaproof}

Lemma~\ref{thm:general} shows that the main part of the expected payment can be expressed by $\delta$, and shows the reason why it incentivizes the agents to deviate in a synchronous way with alternating signs. Lemma~\ref{thm:parity} is a direct consequence of it (using the fact that $\E({\cred p^{red}_k} \cdot {\cblue p^{blue}_k}) = {\cred p^{red}_0} \cdot {\cblue p^{blue}_0}$).

\begin{Lemma} \label{thm:general}
In Example~\ref{general-example},
\begin{alignat}{1}
    \E \Big(\sum_{t=1}^k {\cblue \gamma^{blue}_t}\Big)
    &= 100 \cdot \E \bigg( \Big(  \sum_{t=1}^{k-1} \big({\cred \delta^{red}_{t}} - {\cred \delta^{red}_{t-1}} \big) {\cblue \delta^{blue}_{t}} \Big) - {\cred \delta^{red}_{k-1}} {\cblue \delta^{blue}_k} + {\cred p^{red}_0} {\cblue p^{blue}_0} - {\cred p^{red}_k} {\cblue \widehat{p}^{blue}_k} \bigg) \label{eq:Egb0}
    \\ \E \Big(\sum_{t=1}^k {\cred \gamma^{red}_t}\Big)
    &= 100 \cdot \E\bigg( \Big( \sum_{t=1}^{k-1} \big({\cblue \delta^{blue}_{t}} - {\cblue \delta^{blue}_{t-1}} \big) {\cred \delta^{red}_{t}} \Big) - {\cblue \delta^{blue}_{k-1}} {\cred \delta^{red}_k} + {\cblue p^{blue}_0} {\cred p^{red}_0} - {\cblue p^{red}_k} {\cred \widehat{p}^{red}_k}
    \bigg)
    \label{eq:Egr0}
\end{alignat}
These are true even if players can observe the past types of each other. We assume only that ${\cblue \widehat{p}^{blue}_k}$ is independent of ${\cred p^{red}_k}$ conditional on the history until round $k-1$, and vice versa.
\end{Lemma}

\begin{proof}
The martingale property of probabilities implies that $\E({\cblue p^{blue}_{t}}) = {\cblue p^{blue}_{t-1}}$ and $\E({\cred p^{red}_t}) = {\cred p^{red}_{t-1}}$ for every $t > 0$.
Furthermore, as these martingales for {\cblue Blue} and {\cred Red} are independent, $\E({\cblue p^{blue}_{t_1}} \cdot {\cred p^{red}_{t_2}}) = {\cblue p^{blue}_0} \cdot {\cred p^{red}_0}$ for every $t_1$ and $t_2$.
Similarly, ${\cred p^{red}_k} - {\cred p^{red}_{k-1}}$ is independent of ${\cblue \widehat{p}^{blue}_{k-1}}$, and we assumed that it is also independent of ${\cblue \widehat{p}^{blue}_k}$ and vice versa. These imply the followings.
\begin{equation} \label{eq:exp0b}
    \E\Big( \big( {\cred p^{red}_t} - {\cred p^{red}_{t-1}} \big) \cdot {\cblue \widehat{p}^{blue}_{t-1}} \Big)
    = \E\Big( \big( {\cred p^{red}_t} - {\cred p^{red}_{t-1}} \big) \cdot {\cblue \widehat{p}^{blue}_t} \Big)
    = \E\Big( \big( {\cred p^{red}_t} - {\cred p^{red}_{t-1}} \big) \cdot {\cblue \delta^{blue}_{t-1}} \Big) = 0
\end{equation}
\begin{equation} \label{eq:exp0r}
    \E\Big( \big( {\cblue p^{blue}_t} - {\cblue p^{blue}_{t-1}} \big) \cdot {\cred \widehat{p}^{red}_{t-1}} \Big)
    = \E\Big( \big( {\cblue p^{blue}_t} - {\cblue p^{blue}_{t-1}} \big) \cdot {\cred \widehat{p}^{red}_t} \Big)
    = \E\Big( \big( {\cblue p^{blue}_t} - {\cblue p^{blue}_{t-1}} \big) \cdot {\cred \delta^{red}_{t-1}} \Big) = 0
\end{equation}
If add up \eqref{eq:gb} for every $t$ and we take expectation, then we get the following.
\begin{equation*}
    \E \Big(\sum_{t=1}^k {\cblue \gamma^{blue}_t}\Big)
    = 100 \cdot \sum_{t=1}^k \E \Big( {\cred \widehat{p}^{red}_{t-1}} \cdot
    \big({\cblue \widehat{p}^{blue}_{t-1}} - {\cblue \widehat{p}^{blue}_t}\big) \Big)
\end{equation*}
\begin{equation*}
    = 100 \cdot \sum_{t=1}^k \E \Big(
    {\cred \delta^{red}_{t-1}} \cdot \big({\cblue \delta^{blue}_{t-1}} - {\cblue \delta^{blue}_t}\big)
    + {\cred \delta^{red}_{t-1}} \cdot \big({\cblue p^{blue}_{t-1}} - {\cblue p^{blue}_t}\big)
    + {\cred p^{red}_{t-1}} \cdot \big({\cblue \widehat{p}^{blue}_{t-1}} - {\cblue \widehat{p}^{blue}_t}\big) \Big)
\end{equation*}
\begin{equation*}
    \mathop{=}^{\eqref{eq:exp0r}} 100 \cdot \sum_{t=1}^k \E \Big(
    {\cred \delta^{red}_{t-1}} \cdot \big({\cblue \delta^{blue}_{t-1}} - {\cblue \delta^{blue}_t}\big)
    + {\cred p^{red}_{t-1}} \cdot \big({\cblue \widehat{p}^{blue}_{t-1}} - {\cblue \widehat{p}^{blue}_t}\big) \Big)
\end{equation*}
\begin{equation*}
    = 100 \cdot \sum_{t=1}^k \E \Big(
    {\cred \delta^{red}_{t-1}} \cdot {\cblue \delta^{blue}_{t-1}} - {\cred \delta^{red}_{t-1}} \cdot {\cblue \delta^{blue}_t}
    + {\cred p^{red}_{t-1}} \cdot {\cblue \widehat{p}^{blue}_{t-1}} - {\cred p^{red}_{t-1}} \cdot {\cblue \widehat{p}^{blue}_t}\big) \Big)
\end{equation*}
\begin{equation*}
    \mathop{=}^{\eqref{eq:exp0b}} 100 \cdot \E \Big(
    \sum_{t=1}^{k-1} {\cred \delta^{red}_t} \cdot {\cblue \delta^{blue}_t} - \sum_{t=1}^k {\cred \delta^{red}_{t-1}} \cdot {\cblue \delta^{blue}_t}
    + \sum_{t=1}^k \big( {\cred p^{red}_{t-1}} \cdot {\cblue \widehat{p}^{blue}_{t-1}} - {\cred p^{red}_t} \cdot {\cblue \widehat{p}^{blue}_t}\big) \Big) 
\end{equation*}
\begin{equation*}
    100 \cdot \E \bigg( \Big(  \sum_{t=1}^{k-1} \big({\cred \delta^{red}_{t}} - {\cred \delta^{red}_{t-1}} \big) {\cblue \delta^{blue}_{t}} \Big) - {\cred \delta^{red}_{k-1}} {\cblue \delta^{blue}_k} + {\cred p^{red}_0} {\cblue p^{blue}_0} - {\cred p^{red}_k} {\cblue \widehat{p}^{blue}_k} \bigg)
\end{equation*}

This proves \eqref{eq:Egb0}, and we can get the proof of \eqref{eq:Egr0} in an analogous way.
\end{proof}

\section{Detailed analysis of Example~\ref{general-example}} \label{App1}

We try to analyze Example~\ref{general-example} with small numbers $K$ of rounds and assume that the agents are always able to report arbitrary probabilities in $[0, 1]$. Keep in mind that the message of the paper applies better to larger $K$. One of the purposes of this section is to show the reason why we did not give a full characterization of the PBEs in Example~\ref{general-example}.

\bigskip

{\bf If} $\boldsymbol{K = 1}$, then PBE only means Nash equilibrium, and we play a variant of the coordination game (or battle of sexes). There are 2 stable and 2 or 3 unstable BPEs depending on ${\cblue p^{blue}_0}$ and ${\cred p^{red}_0}$.
\begin{enumerate}
\item The agents always report ${\cblue \widehat{p}^{blue}_1} = {\cred \widehat{p}^{red}_1} = 0$.
\item If ${\cblue p^{blue}_0}, {\cred p^{red}_0} \le 84\%$, then another equilibrium is always reporting ${\cblue \widehat{p}^{blue}_1} = {\cred \widehat{p}^{red}_1} = 1$.
\item The agents always report the truth, namely, ${\cblue \widehat{p}^{blue}_1} = {\cblue p^{blue}_1} \in \{{\cblue 0}, {\cblue 1}\}$ and
${\cred \widehat{p}^{red}_1} = {\cred p^{red}_1} \in \{0, 1\}$.
\item If ${\cblue p^{blue}_1} = {\cblue 0}$, then ${\cblue \widehat{p}^{blue}_1} = {\cblue 0}$. If ${\cblue p^{blue}_1} = {\cblue 1}$, then $\P({\cblue \widehat{p}^{blue}_1} = {\cblue 1}) = \frac{100}{104}$. The same applies for {\cred Red}.
\item Only for ${\cblue p^{blue}_0}, {\cred p^{red}_0} \le 84\%$. If ${\cblue p^{blue}_1} = {\cblue 1}$, then ${\cblue \widehat{p}^{blue}_1} = {\cblue 1}$. $\P\big(({\cblue p^{blue}_1} = {\cblue 0})\text{ AND }({\cblue \widehat{p}^{blue}_1} = {\cblue 1})\big) = \frac{16}{84} {\cblue p^{blue}_0}$. The same applies for {\cred Red}.
\end{enumerate}

We note that if cheap talk is allowed, then the set of PBEs is richer.
It can be even richer with the possibility of some noisy signaling or other intermediate levels of communication.

Even though $K = 1$ used a binary type and a binary report per agent, it already had a number of PBEs. Now it is easier to believe that if $\boldsymbol{K \ge 2}$, where the first signal and the first report are both an arbitrary real number from $[0, 1]$, then the set of PBEs is very rich, we are not able to characterize it.

From an applied point of view, there are two reasonable PBEs. One is ${\cblue \widehat{p}^{blue}_1} = {\cred \widehat{p}^{red}_1} = 0$, ${\cblue \widehat{p}^{blue}_2} = {\cred \widehat{p}^{red}_2} = 1$. The other one is ${\cblue \widehat{p}^{blue}_1} = {\cred \widehat{p}^{red}_1} = 1$, ${\cblue \widehat{p}^{blue}_2} = {\cred \widehat{p}^{red}_2} = 0$.
The truthful strategy profile is one of the unstable PBEs, which does not respect weak dominance because after ${\cblue \widehat{p}^{blue}_{K-1}} = {\cblue 0}$, {\cred Red} is weakly better by reporting ${\cred \widehat{p}^{red}_K} = {\cred 1}$.

\section{The meaning of ``type'' in the Athey--Segal paper} \label{sec:type}

This section is a more detailed analysis extending Section~\ref{sec:type-info} about the possible meanings of ``type'' in the Athey-Segal paper, and how it could have been defined.

\subsection{If ``type'' means payoff type}

The example of Athey and Segal in \cite{AtSe} (Section 3) may suggest that type means payoff type.
However, this is NOT the understanding of that paper. For example, independent payoff types would make no restriction about the information of the players, but it contradicts the following quote.

\begin{center}
``independent types (...) means that, conditional on decisions
\\ (...), an agent's private information does not have any effect
\\ on the distribution of the current (...) types of other agents''
\end{center}

\noindent
We note that the word ``effect'' is a bit misleading here. For example, in their sense, if $i$ knows $\theta^j_1$, then his knowledge ``has an effect'' on the distribution of $\theta^j_t$.



However, the paper and the results could be modified so that ``type'' would mean payoff type, and they could have proved the same theorem with essentially the same proof.
Because by the same reason as we explained in Appendix~\ref{sec:belief}, if we allowed the agents to observe the past of each other to any degree, then formally this would not rule out the ``unconvincing'' PBE.


\subsection{If ``type'' means full type}

In this case, $\theta^i_{t+1}$ includes all public reports from the previous round $\widehat{\theta}^N_t$. This contradicts that it is chosen by a stochastic function $\mu(x_t^i, \theta_t^i)$, because $\widehat{\theta}^N_t$ cannot be written as a function of $x_t^i$ and $\theta_t^i$.

\subsection{If ``type'' means full type excluding the listed public information}

This understanding looks unnatural, but it is consistent with the paper.
However, with this understanding, the independent types assumption still means that this is a very restrictive model.
Formally, it does not even allow two agents to observe the same weather because it would contradict the assumption of independent types.
To be fair, we know from different arguments that this alone cannot spoil a PBE-implementation.
But the model and the PBE-implementation definitely do not allow an agent to reveal any information that is correlated with his past payoff type.

For example, assume that {\cblue Blue} can only have two types, {\cblue HIGH} or {\cblue LOW}, and his type is constant throughout the game.
Assume that {\cblue Blue} has {\cblue LOW} type but {\cred Red} believes it is {\cblue HIGH}. Then whatever {\cblue Blue} tries to reveal, he cannot shake the faith of {\cred Red} in his wrong belief, even in the weakest sense: {\cblue Blue} cannot induce {\cred Red} to make a dominating move which gives {\cred Red} an extra utility 1 if {\cblue Blue} has type {\cblue LOW}, but it makes no difference if his type is {\cblue HIGH}.
Specifically, this assumption means that even if {\cblue Blue} reports {\cblue HIGH} and {\cblue LOW} alternately throughout the game, he cannot prevent {\cred Red} from always being sure that the latest report is the truth.

\subsection{A pure example for an  ``unconvincing'' equilibrium} \label{sec:simple-example}

We give a rough definition of an ``unconvincing'' implementation using the following example.

\begin{Example} \label{simple-example}
Consider the following game with $k$ rounds and with $n \ge 2$ agents. In every round, every agent chooses YES or NO simultaneously. For each round and each agent $i$, if $i$ chooses YES and at least one other agent chose YES in the current or in any of the previous rounds, then $i$ gets utility 1, otherwise, he gets 0 for that round. Every agent maximizes his (expected) total utility of the $k$ rounds.
\end{Example}

We did not specify the information of the players (observability of past actions, signaling about them, cheap talk, etc.), so in this sense, this example describes a class of games. The most important version is when the decisions are private information. But from a more applied point of view, we should keep an eye on the cases where the decisions are not perfectly private.

\medskip

If we use our intuitive understanding of rational behavior, then we see that always choosing YES is the best strategy, and choosing NO has only disadvantages. 
$k = 1$ round is already an interesting case with (at least) two Nash equilibria. $(NO, NO, ..., NO)$ is the other Nash equilibrium which we call the ``unconvincing'' equilibrium.
Always choosing NO remains a PBE with a larger number of rounds $k \ge 2$, with appropriate beliefs.
We can say that this PBE is ``even less convincing'' than with $k=1$.
Because even after an agent $i$ chose YES, and $i$ would like to share the fact that he chose YES, this equilibrium requires every other agent $j$ to keeping believe that $i$ chose NO.
(More precisely, this ``unconvincing'' PBE is formally valid even if the agents observe the decisions of each other, see Appendix~\ref{sec:belief}.)

If a mechanism implements a goal in such an equilibrium, then we say that this is an  ``unconvincing'' implementation.
Especially if such an issue happens not only for degenerate setups, but the mechanism transforms completely generic setups into degenerate games and implements a goal with an ``unconvincing'' equilibrium.

\section{Is PBE a refinement of subgame-perfect equilibrium?} \label{sec:belief}



PBE is a refinement of subgame-perfect equilibrium, but it applies only for proper subgames, meaning that it is common knowledge that they are in this subgame. But if one of the agents $i$ knows it not for sure but only with probability 1, then it is no longer a proper subgame, and subgame-perfection does not apply here.

According to the standard understanding of information and belief, if an observation is a part of the rules, then it is observed for sure. But if we are speaking about a robust model to observability (for example, the Athey--Segal model with dependent types), then an observation may be made with probability 1 but never for sure.
And hereby, if we apply the one-shot deviation principle, then we can find a PBE without noticing that the beliefs of the agents are inconsistent with their observations.

Specifically, the ``unconvincing'' PBE in Example~\ref{simple-example} is formally valid even if the agents partially or perfectly observe the decisions of each other (unless if we add any perfect observation to the \emph{rules} of the game).
Because if every agent believes that everybody else chose and will choose NO, then if agent $i$ observes any YES by another agent $j$, then it is a 0-probability event according to the belief of agent $i$. Therefore, the Bayesian update rule allows agent $i$ keeping believe that every other agent chose and will choose NO (and agent $i$ believes that his eyes just dazzled when he observed a YES from $j$), and agent $i$ can believe that this is still a common belief (with probability 1).


\section{Less reporting does not fix the problem} \label{app:roundsense}

Example~\ref{general-example} uses reports which do not have any effect on the public decision.
One could naively ask whether we can fix the mechanism by always asking to report the smallest necessary information for the efficient decision.

The smallest problem is that this question is not well-defined because there is no canonical definition of the smallest necessary information to be reported for the efficient decision. But the following two arguments show more clearly that these kinds of ideas cannot fix the problem.

{\bf Reason 1.} We can easily construct an example to show that an agent can make a useless report of his type essentially by claiming a fake reason.
Formally, consider an agent and a round where his type $\theta$ does not have to be reported. Let us modify the setup in the following way. The agent receives an extra bit $b \in \{0, 1\}$, and we add an extra public decision $\widehat{\theta}$. If $b = 1$ and $\widehat{\theta} = \theta$, then the agent gets an extra utility $\epsilon$. If $b = 0$, then $\widehat{\theta}$ has no effect.

Clearly, if $b = 1$, then the agent should report his type.
But if $b = 0$ and he wants to, then he can still claim that $b = 1$ and hereby he can report his type.
This kind of technique works unless if the designer can rule out every possibility that the agent may make use of a public decision that does not affect the others.

{\bf Reason 2.} Even a tiny reason can make it necessary for a type to be reported. But from a practical point of view, our counterexample in Section~\ref{sec:wrong-ex} has some robustness. It may be useful to go back to Example~\ref{simple-example} and see what if we introduce an arbitrarily small cost for choosing YES. It would make always choosing NO a formally legit PBE in many senses. Mainly because it will no longer be weakly dominated by any other strategy. However, from a more applied point of view, this equilibrium would still be ``pretty unconvincing'', especially if $k$ is large. Now if we compare it with Example~\ref{general-example}, then we can see that the situation is very similar. Therefore, we cannot hope that such a modification could really fix the problem.

\section{Weaknesses of the unbalanced Team Mechanism} \label{unbalanced}

The unbalanced Team Mechanism also has a weakness in that two colluding agents can get as much payoff as they want. We show an example of this issue.

Consider the following setup with two agents. In each round, each agent receives a signal (payoff type) 1000 or -1 independently, with probabilities of $1/2$. Then the designer makes a public decision YES or NO. If it is NO, then both agents get utility 0. But if the decision is YES, then both agents get the utility equal to the signal.

Consider the case when an agent receives a signal $-1$. If he reports $1000$ instead, then it costs him $0$ or $2$, but it provides $1001$ or $999$ more utility to the other agent. The former amounts (costing $0$ providing $1001$) apply if the agent expects the other agent to report $1000$.
This is a strong motivation for collusion, moreover, it does not even require collusion, especially in the infinite-horizon game.
For example, reporting $1000$ as long as the other agent also did so (since the beginning, or only in the previous round, or in the previous $k$ rounds, any of these versions work) is another PBE. This equilibrium is not efficient, but it provides a higher payoff for the two agents.

\section{How could a valid argument lead to an ``unconvincing'' mechanism?} \label{how}

The issue primarily arises from the known weaknesses of Perfect Bayesian Equilibrium and the standards of non-robust mechanism design. 


PBE is a (slightly imperfect) necessary condition of the intuitive term ``rationality'', but it is not at all a sufficient condition.
Therefore, for analyzing games, it is a meaningful plan to find every PBE of the game, and after that, we make an analysis in which we understand which of the PBEs are meaningful for practice.
This latter analysis may not be completely well-defined but it may use our intuitive understanding of rationality.
It does not sound ideal, but this is the best we can do. 
Accordingly, if we design a mechanism to implement a goal in a practically meaningful way, then it is not enough at all to check that, under the mechanism, one of the PBEs satisfies the goal.

In addition, the PBE-implementation of the Athey--Segal mechanism used the one-stage deviation principle. It created an even higher risk of implementing the goal with an ``unconvincing'' PBE because of the following reason.
PBE has a blind spot about off-equilibrium paths.
The Bayesian update rule gives no restriction on the belief of an agent after an off-equilibrium action of another agent.
Every belief is allowed here not because every belief is rational, but the definition of PBE does not undertake to tell what cannot be a rational belief in this situation.
(PBE is so much permissive that it does not always respect subgame-perfection, see Appendix~\ref{sec:belief}.) 
The one-stage deviation principle (with the revelation principle) uses this freedom of beliefs in a PBE to avoid the rationality conditions about the possibility that another player already deviated.
In other words, the one-stage deviation principle is functioning also as a hacking technique that shifts some of the rationality conditions into this blind spot of PBE.
Applying the one-stage deviation principle to Example~\ref{simple-example} can help to understand this issue better.

A further important weakness of the balanced Team Mechanism is that it strictly assumes an unrealistic information structure, as we discussed it in Section~\ref{sec:weaknesses}.

\section{The original mechanism and nonquasilinear payoffs} \label{app:non-qlinear}

The mechanism in \cite{CsE} was designed for a tendering setup where a number of agents are competing for participation in a project owned by another player called the principal. Rejected agents will do nothing and gain nothing.
Now the mechanism was the following.
Each agent makes an offer for a contract about the payment rule between him and the principal, as a function of the contractible events including communication.
If an agent is accepted, then this payment rule will apply.
Then the principal chooses a strategy (about which of them to accept and how to communicate with the agents) that maximizes her minimum possible expected payoff, where expectation applies only to the stochastic changes on her own private type and the public type.

When we apply that original mechanism to this model with fixed initial types, then it means that we assume that the principal and the agents reported truthfully about their initial types, and we consider the continuation of the game with the winner agents of this tender and with the principal.
If we assume that the agents have quasilinear payoffs (i.e. payoff = utility + payment), and we restrict the set of reports to the potentially truthful ones, then we get the mechanism in Section~\ref{sec:my-mechanism}.

The original mechanism does not have any implicit assumption of
quasilinear payoffs, and accordingly, that version works ``better'' if some agents have nonquasilinear payoffs (or they have limited responsibility, etc.).
It is very difficult to formally compare two mechanisms partially because the efficient outcome is no longer well-defined here.
Therefore, we will just argue in an intuitive way that this is a nicer and more natural generalization of the direct mechanism.
For example, if a public decision has an effect only on agent $i$ about whether or not to take a huge risk, then even though $i$ might be very risk-averse, the mechanism in Section~\ref{sec:my-mechanism} (or the Team Mechanisms) makes the decision only considering the expected gain. (In simple cases, the agent can hack the mechanism by misreporting, but we cannot get an universal solution by individual hackings.)
But the original mechanism chooses the best according to the preference of $i$.

In more detail, assume that the agents have a payoff function $F(\theta_{1, 2, ..., T}^{0,i}, x_{1, 2, ..., T}, y^i)$ which is monotonic and unbounded by $y^i$, and every agent maximizes the expected value of $F$. 
If we restrict the set of reports to the potentially truthful (or recommended) ones with this weaker assumption, then the mechanism will have the following structure.

Given the types, we determine a decision strategy, the payment rule $y^i$ (with $\sum_{i\in N} y^i \equiv 0$) and the expected payoffs $f^i$ satisfying the following.
For each agent $i$, if he reports his types and makes his private decisions truthfully (implying $\theta_{1, 2, ..., T}^{0,i} = \widehat{\theta}_{1, 2, ..., T}^{0,i}$), then 
\begin{equation*}
\widehat{\theta}_{1, 2, ..., T}^{N \setminus \{i\}} \to \E_{\theta_{1, 2, ..., T}^{0,i}}\Big(F\big(\theta_{1, 2, ..., T}^{0,i},\ x_{1, 2, ..., T}(\widehat{\theta}_{1, 2, ..., T}^{N_0}),\ y^i(\widehat{\theta}_{1, 2, ..., T}^{N_0})\big)\Big) \equiv f^i,
\end{equation*}
where the distribution of $\theta_{1, 2, ..., T}^{0,i}$ is also affected by $\widehat{\theta}_{1, 2, ..., T}^{N \setminus \{i\}}$ via $x_{1, 2, ..., T}$, and $f^i$ is a constant.

These properties define the mechanism by a backward recursion in the sense that we can recursively define the vectors $f^N$ we can achieve.
The point is that each time an agent $i$ reports a change in his private type, he compensates the other agents for the effect, making them neutral about this change. Therefore, if we replace the payoff function of $i$ with its expected value before the last change in his private type considering the induced payments, then we get to the same problem with one round less.
Reduction by one of the latest decisions is simpler, we just take the union of the sets of $f^N$ with the different decisions.

It may seem that there is an additional problem here with the choice between the different Pareto optimal vectors $f^N$ (at the beginning of the game).
But the truth is that this is essentially the same problem as we had and ignored in the case of quasilinear payoffs.
Namely, with quasilinear payoffs, we can add an arbitrary additional balanced transfer rule between the agents depending on their initial types, and constant 0 is not a fair choice in practice.
For example, consider the deterministic setup with two agents and a binary public decision to be made in every round, where {\cblue Blue} always prefers YES by 100 and {\cred Red} always prefers NO by 101. Clearly, every decision should be NO, but any applied mechanism would probably ask {\cred Red} to pay some compensation for {\cblue Blue}.
So the only special property of the case of quasilinear payoffs is that the problems with efficiency and fairness can be split \emph{additively} into two independent problems because the Pareto optimal expected payoffs always form a hyperplane.

We note that in the setup in \cite{CsE} with competing agents, this specialty does not matter. Essentially because the outside option is a reference point for payoffs or because each agent $i$ submits an offer for a contract including a required expected payoff $f^i$.

\end{document}